\documentclass[aps,prl,twocolumn,superscriptaddress,floatfix,nofootinbib,showpacs,longbibliography,groupedaddress]{revtex4-1}

\usepackage{xcolor}
\usepackage[utf8]{inputenc}  
\usepackage[T1]{fontenc}     
\usepackage[british]{babel}  
\usepackage[sc,osf]{mathpazo}
\usepackage[scaled=0.86]{berasans}  
\usepackage[colorlinks=true, citecolor=blue, urlcolor=blue]{hyperref}  
\usepackage{graphicx} 
\usepackage[babel]{microtype}  
\usepackage{amsmath,amssymb,amsthm,bm,amsfonts,mathrsfs,bbm} 

\usepackage{xspace}  
\usepackage{xcolor}
\usepackage{multirow}
\usepackage{array}
\usepackage{bigstrut}
\usepackage{braket}
\usepackage{color}
\usepackage{natbib}
\usepackage{multirow}
\usepackage{mathtools}
\usepackage{float}
\usepackage[caption = false]{subfig}
\usepackage{xcolor,colortbl}
\usepackage{color}

\newcommand{\be}{\begin{equation}}
\newcommand{\ee}{\end{equation}}
\newcommand{\ba}{\begin{eqnarray}}
\newcommand{\ea}{\end{eqnarray}}

\newcommand{\tr}{\operatorname{Tr}}

\newtheorem{theorem}{Theorem}

\def\>{\rangle}
\def\<{\langle}

\begin{document}
\title{Reversible computation and the causal structure of space-time}

\author{Anandamay Das Bhowmik}
\affiliation{Physics and Applied Mathematics Unit, Indian Statistical Institute, 203 BT Road, Kolkata- $700108$, India.}

\author{Preeti Parashar}
\affiliation{Physics and Applied Mathematics Unit, Indian Statistical Institute, 203 BT Road, Kolkata- $700108$, India.}

\begin{abstract}
Reversible algorithms play a crucial role both in classical and quantum computation. While for a classical bit the only nontrivial reversible operation is the bit-flip, nature is far more versatile in what it allows to do to a quantum bit. The reversible operations that a quantum computer can perform on a qubit are group of linear unitary transformations. However, laws of quantum mechanics prohibit implementation of anti-linear anti-unitary gates, even though they are perfectly reversible. Here we show that such a restriction on possible set of reversible operations is, remarkably, a fundamental constraint of spacetime structure. In particular, it will be shown that construction of any anti-linear anti-unitary gate will lead to violation of a fundamental causal primitive which, as we shall argue, is fundamentally different from the principle of relativistic causality.

\end{abstract}

\maketitle
{\it Introduction --} Two reversible operations one can perform on a classical bit - leaving it alone or flipping it. However, possibilities for reversible operations get richer when we go from a classical bit to a quantum bit. In fact, efficiency of quantum computational algorithms does heavily rely on reversible manipulation of  qubits. It is, in principle, possible to transform any generic irreversible computation into a reversible one and perform complex computation by means of reversible gates \cite{note, Nielsen}. However, the sole condition of reversibility does not imply that evolution must necessarily be linear unitary. Any invertible transformation that maps unit vectors of Hilbert space into unit vectors  should, in principle, be a potential candidate in the set of reversible operations and indeed it follows from the celebrated Wigner's theorem that such a set is large enough to contain both linear unitary and anti-linear anti-unitary transformations \cite{Wig,Peres1, Weinberg}. However, the dynamical postulates of quantum theory only allows the group of linear unitary transformations to be physically realizable. Despite difficulties in practical implementation, it is possible to physically realize a linear unitary transformation to an arbitrary degree of accuracy. However on the other hand, perfect implementation of an antilinear anti-unitary logic gate is impossible to achieve since they are not completely positive. This fundamental asymmetry between realizability of linear unitary and antilinear anti-unitary logic gates, though inherent in the mathematical formalism of quantum dynamics, is not clearly understood. Peres had argued to rule out the possibility of anti-linear anti-unitary evolution  from the requirement of continuity \cite{Peres2}, however the rational behind the assumption that physical evolution must necessarily belong to the continuous group of transformations requires more fundamental justification.

In this work, we show that perfect implementation of any anti-linear anti-unitary gate will lead to  violation of a fundamental causal primitive. (importantly, this causal primitive is different from the principle of relativistic causality). In brief, the causal primitive states that different possible time orders of any two spacelike separated operations, being reference frame dependent, cannot become a potential cause for distinct observations in their common causal future. In other words, it asserts that cause of any event cannot be observer dependent \cite{us}. Since in our proof we do not invoke any ingredient from the mathematical apparatus of Schrödinger dynamics, our result establishes that the impossibility of physical realization of antilinear anti-unitary gates is not just a mere artefact
of the current formalism of quantum dynamics, rather it is intimately connected with the causal structure of spacetime.

{\it The causal primitive --} The statement of the causal primitive can be formulated as the following : {\it If $A$ is the cause of $B$ for some observer, then $A$ will remain as the cause of $B$ for all observers}. Therefore, it demands that the cause $A$ of any event $B$ must be reference frame independent {\it i.e.} absolute. From now on, this causal primitive will be referred to as 'Absoluteness of Cause' (AC). It is worth mentioning here that although cause is usually associated with events in spacetime or operations performed in sufficiently small region of spacetime, in scenarios that involve multiple operations, the spacetime relation between those operations (for instance, their `time order') should also be regarded as a potential cause. To see that, consider two operations $\mathrm{E}_a(t_a,{\bf r}_a)$ and $\mathrm{E}_b(t_b,{\bf r}_b)$ carried out on a rigid body. Say $\mathrm{E}_a(t_a,{\bf r}_a)$ is a finite rotation performed on the rigid body about x-axis and $\mathrm{E}_b(t_b,{\bf r}_b)$ is a finite rotation performed on the same body about y-axis. The possible temporal orders are : either $t_a<t_b$ or $t_b<t_a$. Since finite rotations along different axis do not commute in general, these two different temporal orders of $\mathrm{E}_a(t_a,{\bf r}_a)$ and  $\mathrm{E}_b(t_b,{\bf r}_b)$ will cause distinct future configurations. Therefore, time order of two operations can be a potential cause for distinct future observations. Since in this example both the operations belong to the world-line of the same physical system, $\mathrm{E}_a(t_a,{\bf r}_a)$ and $\mathrm{E}_b(t_b,{\bf r}_b)$ are timelike separated and hence their time order is observer independent, and therefore, in this case, AC does not forbid this time order to be a cause for distinct future observations (since time order in this case is observer-independent). The situation is however different if the operations are spacelike separated. Consider $\mathrm{E}_a(t_a,{\bf r}_a)$ and $\mathrm{E}_b(t_b,{\bf r}_b)$ to be spacelike separated operations carried out on two spatially separated systems. The time order of such spacelike separated operations is  non-absolute \textit{i.e.} observer dependent - there exist different inertial frames where these operations take place in different time orders. The causal primitive AC prohibits time order of such spacelike separated operations (for not being absolute) to be a potential cause for distinct observations in their common future.

{\it Relativistic Causality versus AC --} 
Quantum theory is compatible with the principle of  AC. To see that, assume Alice and Bob share a bipartite quantum state $\rho_{AB} \in \mathcal{D(H_A \otimes H_B)}$.
Consider two spacelike separated quantum operations $\Lambda_A : \mathcal{D(H_A)}\xrightarrow{} \mathcal{D(H'_A)}$ and $\Lambda_B : \mathcal{D(H_B)}\xrightarrow{} \mathcal{D(H'_B)}$ on system A and B respectively \cite{note1}. Alice performs $\Lambda_A$ on her part of $\rho_{AB}$, whereas Bob performs $\Lambda_B$ on his part. Since to perform these local operations, they do not need to communicate with each other, $\Lambda_A$ and $\Lambda_B$ can be spacelike separated. The temporal order between them is observer dependent. The principle of AC prohibits the possible time orders of such spacelike separated operations to be a potential cause of distinct observations at their common future. Consider two possible time order of Alice and Bob's local operations.
\begin{itemize}
\item[1.] Alice performing $\Lambda_A$ on her system precedes Bob performing $\Lambda_B$ on his system. The joint state $\rho_{AB}$ will evolve to $(\mathbb{I}_A\otimes \Lambda_B)(\Lambda_{A} \otimes \mathbb{I}_B)\rho_{AB} = \tau_{AB}$.
\item[2.] Bob performing $\Lambda_B$ on his part precedes Alice performing $\Lambda_A$ on her part. In this case, the joint state $\rho_{AB}$ evolves to $(\Lambda_{A} \otimes \mathbb{I}_B)(\mathbb{I}_A\otimes \Lambda_B)\rho_{AB} = \tau'_{AB}$
\end{itemize}
However, Quantum theory ensures that $\tau_{AB} = \tau'_{AB} ;~\forall~ \rho_{AB} \in \mathcal{D(H_A \otimes H_B})$, since $
(\mathbb{I}_A\otimes \Lambda_B)(\Lambda_{A} \otimes \mathbb{I}_B) = (\Lambda_A \otimes \Lambda_B) = (\Lambda_{A} \otimes \mathbb{I}_B)(\mathbb{I}_A\otimes \Lambda_B)$.
Therefore, through commutation of $(\mathbb{I}_A\otimes \Lambda_B)$ and $(\Lambda_{A} \otimes \mathbb{I}_B)$, quantum theory ensures that different time-orders of any two spacelike quantum operations does not lead to different final joint states. Hence time-order of two spacelike quantum operations cannot cause distinct observations in their common future, thus respecting the principle of AC.

It is well-known that despite allowing nonlocal correlations between spatially separated subsystems, quantum theory maintains 'peaceful co-existence' with the principle of relativistic causality (RC) \cite{Brunner}. However, AC and RC are not equivalent. To see how quantum theory obeys RC, consider again a situation where Alice and Bob share a bipartite quantum state $\rho_{AB}$. The marginal density matrix of Bob's system is given by $Tr_A(\rho_{AB}) : = \mu_B$. Now we assume
that Alice has performed a local quantum operation $\Lambda_A$ on her subsystem of the composite state $\rho_{AB}$. As a result of Alice's local operation $\Lambda_A$, the joint state $\rho_{AB}$ will evolve to $ \rho'_{AB} = (\Lambda_{A} \otimes \mathbb{I}_B) \rho_{AB} = \sum_k A_k \otimes \mathbb{I}_B (\rho_{AB}) A_k^{\dagger} \otimes \mathbb{I}_B.$ where $\{A_k\}$ are the Kraus elements corresponding to $\Lambda_A$. The state of Bob's subsystem after Alice's operation is given by $\mu'_B = Tr_A(\rho'_{AB})$. If $\mu'_B \neq \mu_B$, Alice can instantaneously signal to B, thus violating RC. However, quantum theory ensures $\mu'_B = \mu_B$, thus prohibiting instantaneous communication. $\mu'_B  = Tr_A(\sum_{k} A_k \otimes \mathbb{I}_B ( \rho_{AB} ) A_k^{\dagger} \otimes \mathbb{I}_B) = Tr_A (\sum_k A_k^{\dagger} A_k \otimes \mathbb{I}_B (\rho_{AB})) = Tr_A(\rho_{AB}) = \mu_B$.
Therefore, the reduced density matrix of Bob's system remains unaffected by any local operation carried out by Alice on her system, thus maintaining RC. The violation of AC does not necessarily imply violation of RC (see supplemental for more detail). In the next two sections we will show that antilinear anti-unitary logical manipulation of qubits are in direct contradiction with AC while completely compatible with RC. 

{\it Antilinear anti-unitary logic and RC --}
We consider the following question: can the allowed
set of reversible quantum gates be derived from the principle of relativistic
causality? In particular, we ask, whether or not, existence of perfect antilinear anti-unitary gates leads to violation of RC, so that their existence can be ruled out not only from mathematical formalism of quantum dynamics but also from theory-independent generic physical principle such as RC. Here we answer this question in the negative - although it is well known that RC can rule out several unphysical information processing tasks ( for instance, No-cloning, No-deleting etc can be derived as a consequence of RC \cite{Hardy, Gisin, Pati,Bey,Chattopadhyaya06}). Consider the generic bipartite scenario with Alice and Bob, where Alice and Bob are two distant parties sharing some entangled state $\rho_{AB}$. The reduced density matrix at Bob's site is the mixed state $Tr_A(\rho_{AB}) := \rho_B$. Now one crucial non-classical feature of quantum mechanics is that it allows a mixed state to be prepared in many different ways \textit{i.e.} as mixtures of different ensemble of pure states. So $\rho_B$
admits many decomposition as follows : $\rho_B = \sum_i p_i \ket{\psi_i}\bra{\psi_i} = \sum_j q_j \ket{\phi_i}\bra{\phi_i}$. Furthermore, due to GHJW theorem, these different decompositions of $\rho_B$ can be prepared remotely by Alice by choosing appropriate measurements on her local system $A$ \cite{GHJW}. Suppose if Alice performs $M_1$ or $M_2$ measurement on $A$, at Bob's site the decomposition $\sum_i p_i \ket{\psi_i}\bra{\psi_i}$ or  $\sum_j q_j \ket{\phi_i}\bra{\phi_i}$ will be prepared respectively.
Now consider Bob possesses a device $\mathcal{E}$ that performs some operation  $\mathcal{E} : \ket{\psi}\bra{\psi} \xrightarrow[]{}\mathcal{E}({\ket{\psi}\bra{\psi}})$. Now if Alice measures $M_1$, then after Bob's application of $\mathcal{E}$ on his system, the final mixture at Bob's end happens to be
$\sum_i p_i \mathcal{E}(\ket{\psi_i}\bra{\psi_i}) := \sigma_1$. On the other hand, if Alice performs $M_2$ measurement, then Bob's application of $\mathcal{E}$ on his system will lead to the final ensemble  $\sigma_2 : = \sum_i q_i \mathcal{E}(\ket{\phi_i}\bra{\phi_i})$. Now if $\sigma_1 \neq \sigma_2$, then Bob can learn which decomposition of $\rho_B$ has been prepared remotely by Alice, therefore the device $\mathcal{E}$ enables Bob to discern which measurement has been chosen by Alice, which implies violation of RC. Now consider $\mathcal{E}$ to be a generic antilinear antiunitary gate. We need to check whether such gates can be used by Bob to learn which decomposition has been prepared remotely by Alice. Here we show that such gates are of no use for Bob to learn which particular decomposition of $\rho_B$ Alice has remotely prepared : A generic anti-linear anti-unitary gate $\mathcal{E}$ can be written as $\mathcal{E} = \mathcal{U}\mathcal{K}$ where $\mathcal{U}$ is some linear unitary operator and $\mathcal{K}$ is complex conjugation operator defined by $\mathcal{K} \ket{\psi} = \ket{\psi}^*$, where $\ket{\psi}^*$ is complex conjugate of $\ket{\psi}$. Now consider any two arbitrary decompositions of Bob's density matrix $\rho_B$ as follows : $\rho_B = \sum_i p_i \ket{\psi_i}\bra{\psi_i} = \sum_j q_j \ket{\phi_j}\bra{\phi_j}$. Bob can discern which of these ensembles has been remotely prepared by Alice if the final ensembles obtained after Bob applies $\mathcal{E}$ on his system yields different measurement statistics. The first ensemble maps to $\sigma'_B = \sum_i p_i \mathcal{E}(\ket{\psi_i}\bra{\psi_i}) =  \sum_i p_i \mathcal{U}\mathcal{K} \ket{\psi_i}\bra{\psi_i} \mathcal{K}^*\mathcal{U}^{\dagger} = \mathcal{U}\mathcal{K} (\sum_i p_i \ket{\psi_i}\bra{\psi_i}) \mathcal{K}^* \mathcal{U}^{\dagger} = \mathcal{U}\mathcal{K} (\rho_B) \mathcal{K}^*\mathcal{U}^{\dagger}$. Similarly, the second ensemble maps to $\sigma''_B = \sum_i q_i \mathcal{E}(\ket{\phi_i}\bra{\phi_i}) = \mathcal{U}\mathcal{K} (\rho_B) \mathcal{K}^*\mathcal{U}^{\dagger}$. Therefore, $\sigma'_B = \sigma''_B$. Hence with any such operation $\mathcal{E} = \mathcal{U}\mathcal{K}$, Bob cannot discern which decomposition of $\rho_B$ has been prepared by Alice remotely, thereby discarding the possibility of instantaneous communication. \\

{\it Principle of AC and antilinear antiunitary gates :--} We prove a generic theorem that perfect implimentation of any anti-unitary antilinear gate, though does not violate RC, leads to the violation of AC. Before proceeding, we remind that any anti-unitary antilinear gate $\mathcal{E}$ can be decomposed as a product of a linear unitary gate and complex conjugation gate , \textit{i.e.} $\mathcal{E} = \mathcal{U}\mathcal{K}$. Our first theorem establishes the contradiction between existence of the universal complex conjugation gate $K$ and the causal primitive AC.  

\begin{theorem}\label{theo1}
Any dynamical formulation of quantum theory that allows a universal complex conjugation gate necessarily violates the principle of AC.  	
\end{theorem}
\begin{proof}
Consider a bipartite entangled state $\rho_{ab} = \ket{\psi^-}_{ab}\bra{\psi^-}$ shared between two distant parties, say Alice and Bob, where $\ket{\psi^-}_{ab}:= (\ket{0}_a\ket{1}_b-\ket{1}_a\ket{0}_b)/\sqrt{2}$; $\{\ket{0},\ket{1}\}$ be the computational basis states \textit{i.e.} the normalized eigenkets of $\sigma_z$ Pauli matrix. Alice and Bob can prepare this singlet state in their common past and then become spatially separated with their respective subsystems, as considered in seminal Einstein-Rosen-Podolsky (EPR) thought experiment \cite{EPR35}. In contrary to Theorem 1, assume there exists a  device that acts as a universal complex conjugation gate $\mathcal{K}$ and Bob possesses this device, with which we will analyze the following three scenarios. Before proceeding let us explicitly note the action of $\mathcal{K}$ on three basis states : $\ket{i} \xrightarrow[]{\mathcal{K}} \ket{i}$ for $i \in \{0,1\}$; $\ket{x} \xrightarrow[]{\mathcal{K}}\ket{x}$ and $\ket{\bar{x}} \xrightarrow[]{\mathcal{K}}\ket{\bar{x}}$ where $\ket{x} = 1/\sqrt{2}(\ket{0} + \ket{1})$, $\ket{\bar{x}} = 1/\sqrt{2}(\ket{0} - \ket{1})$ and $\ket{y} \xrightarrow[]{\mathcal{K}} \ket{\bar{y}}$ where $\ket{y}:= 1/\sqrt{2} (\ket{0} + i\ket{1})$ and $\ket{\bar{y}}:= 1/\sqrt{2} (\ket{0} - i\ket{1})$.

{\it Scenario- $1$ :} Consider two spacelike separated local quantum operations by Alice and Bob : (1) Alice performs $\sigma_z$ measurement on her part of singlet and (2) Bob performs $\mathcal{K}$ on his subsystem (see Figure $1$ in supplemental). Consider the time-order (A) : in which Alice first performs $\sigma_z$ measurement on her part of the composite system and then Bob applies $\mathcal{K}$ on his part of the system. After Alice's measurment, the singlet state collapses to either $\ket{0}_a\bra{0} \otimes \ket{1}_b\bra{1}$ or to $\ket{1}_a\bra{1} \otimes \ket{0}_b\bra{0}$ with equal probability. Therefore, the post-measurement state $\tau_{ab}$ is
\footnotesize
$$\tau_{AB}=\frac{1}{2}(\ket{0}_A\bra{0}\otimes\ket{1}_B\bra{1}+\ket{1}_A\bra{1}\otimes\ket{0}_B\bra{0}),$$
\normalsize
Bob applies $\mathcal{K}$ on his part after Alice's measurement. If Alice's measurement makes the singlet to collapse into $\ket{0}_A\bra{0} \otimes \ket{1}_B\bra{1}$, then due to action of $\mathcal{K}$ on Bob's side $\ket{0}_A\bra{0} \otimes \ket{1}_B\bra{1}$ remains $\ket{0}_A\bra{0} \otimes \ket{1}_B\bra{1}$ (since $\ket{1} \xrightarrow[]{K} \ket{1}$). Similarly, if the singlet collapses to $\ket{1}_A\bra{1} \otimes \ket{0}_B\bra{0} $, after Bob's application of $\mathcal{K}$ on his part $\ket{1}_A\bra{1} \otimes \ket{0}_B\bra{0} $ remains $\ket{1}_A\bra{1} \otimes \ket{0}_B\bra{0} $. Therefore, the final joint state $\tau'_{ab}$ happens to be equal to $\tau_{ab}$  $\textit{i.e.}$ equal probabilistic mixture of  $\ket{0}_A\bra{0} \otimes \ket{1}_B\bra{1}$ and $\ket{1}_A\bra{1} \otimes \ket{0}_B\bra{0}$
\footnotesize
$$\tau'_{AB}= \tau_{AB} = \frac{1}{2}(\ket{0}_A\bra{0}\otimes\ket{1}_B\bra{1}+\ket{1}_A\bra{1}\otimes\ket{0}_B\bra{0}).$$
\normalsize
Now consider the opposite time ordering of the concerned operations (B) in which Bob's application of $\mathcal{K}$ on $B$ precedes Alice's $\sigma_z$ measurement on $A$. Let after the action of Bob the updated joint state is $\Omega_{ab}$. The definition of $\mathcal{K}$ only provides the action of it on qubit in pure states, but not on a qubit which is entangled to some other system is not known. ( however, if we had asssumed Schrodinger type linearity then it would be straightforward to determine $\mathcal{K}$'s action on an entangled qubit.) Since we do not make any such additional assumption regarding the action of $\mathcal{K}$ when it is applied on qubit entangled with other system, the possibilities for $\Omega_{ab}$ is large. However, due to AC, $\Omega_{AB}$ cannot be arbitrary. It has to satisfy the requirement that after Alice measures $\sigma_z$ on her part of $\Omega_{ab}$, the final joint state should be same as $\tau'_{ab}$ (since, due to AC, time ordering of the spacelike separated events (Alice's $\sigma_z$ measurement and Bob's application of device $\mathcal{K}$) cannot lead to an observable distinction. The solution set for $\Omega_{AB}$ compatible with $AC$ principle is the following : 
\footnotesize
\begin{align}\label{eq1}
\!\begin{aligned}
\mathcal{A}_1 : = \mbox{CH}\left\{\ket{\epsilon_k}\bra{\epsilon_k}\right\},~~~~~~~~~~~~~~~~~~~~~~~~~~~~\\
\ket{\epsilon_k}:=\frac{1}{\sqrt{2}}\left(\ket{0}_A\ket{1}_B+e^{i\alpha_k}\ket{1}_A\ket{0}_B\right)~\&~\alpha_k\in[0,2\pi)
\end{aligned},	
\end{align}
\normalsize
where $\mbox{CH}\{X\}$ denotes the convex hull of a set X. Crucially, note that the solution set $\mathcal{A}_1$ contains all two-qubit states (pure and mixed) that yield perfectly anti-correlated outcomes if both the parties measure in the z-basis \textit{i.e.} performs $\sigma_z$ measurement, and yield completely random outcome locally.

{\it Scenario $2$ :} Consider a scenario similar to the previous one but with the exception that Alice performs $\sigma_x$ measurement. Here also we consider two different time order of Alice and Bob's local operations. In the time order where Bob first applies $\mathcal{K}$ on his part of the composite state $\ket{\psi^-}_{ab}$ and then Alice performs $\sigma_x$ measurement on her part then, AC requires that the final joint state should be same as in the other time order where Alice's $\sigma_x$ measurements precedes Bob's application of $\mathcal{K}$. Due to this requirement imposed by AC, the solution set for $\Omega_{AB}$ turns out to be the following :
\footnotesize
\begin{align}\label{eq2}
\!\begin{aligned}
\mathcal{A}_2 : = \mbox{CH}\left\{\ket{\xi_k}\bra{\xi_k}\right\},~~~~~~~~~~~~~~~~~~~~~~~~~~~~\\
\ket{\xi_k}:=\frac{1}{\sqrt{2}}\left(\ket{x}_A\ket{\bar{x}}_B+e^{i\beta_k}\ket{\bar{x}}_A\ket{x}_B\right)~\&~\beta_k\in[0,2\pi)
\end{aligned}.	
\end{align}
\normalsize   
Here the solution set $\mathcal{A}_2$ contains all the two-quibt states that yield perfectly anti-correlated outcomes in x-basis with local statistics completely random.
The set $\mathcal{A}_1$ and set $\mathcal{A}_2$ have only one common solution which is the state $\ket{\psi^-}_{ab}:=(\ket{0}_a\ket{1}_b-\ket{1}_a\ket{0}_b)/\sqrt{2}$, corresponding to $\alpha_k=\pi$ in (\ref{eq1}) and $\beta_k=\pi$ in (\ref{eq2})

{\it Scenario $3$:} Here Alice's measurement is Pauli $\sigma_y$ measurement. Like the earlier two cases considering both time-orders of Alice's measurement and Bob's application of $\mathcal{K}$, AC allows following possibilities for the state $\Omega_{ab}$ ,
\footnotesize
\begin{align}\label{eq3}
\!\begin{aligned}
\mathcal{A}_3 := \mbox{CH}\left\{\ket{\zeta_k}\bra{\zeta_k}\right\},~~~~~~~~~~~~~~~~~~~~~~~~~~~~\\
\ket{\zeta_k}:=\frac{1}{\sqrt{2}}\left(\ket{y}_A\ket{y}_B+e^{i\delta_k}\ket{\bar{y}}_A\ket{\bar{y}}_B\right)~\&~\delta_k\in[0,2\pi)
\end{aligned},	
\end{align}
\normalsize  
Here the set $\mathcal{A}_3$ contains all two-qubit states that yield perfect correlation in y-basis \textit{i.e.} $\sigma_y$ measurement with local statistics completely random. Impotantly,
the set $\mathcal{A}_3$ does not contain the state $\ket{\psi^-}$ which happens to be the common solution for the earlier two cases. In other words, these three different considerations do not allow any common solution for $\Omega_{ab}$, \textit{i.e.} $\mathcal{A}_1 \cap \mathcal{A}_2 \cap \mathcal{A}_3 =\emptyset $, which implies existence of a universal complex conjugation gate cannot exist if AC is satisfied. This completes the proof.     
\end{proof}

From theorem $1$, we can immediately arrive at our main theorem which establishes that existence of any perfect antilinear anti-unitary gate leads to violation of AC.

\begin{theorem}\label{theo2}
Any dynamical formulation of quantum theory that allows perfect implementation of antilinear anti-unitary gates necessarily violates the principle of AC. 

\end{theorem}

\begin{proof}
In theorem $1$, we have already established that AC rules out the possibility of the universal complex conjugation gate. Now, any antilinear anti-unitary $\mathcal{E}$ transformation is the complex conjugation operation $\mathcal{K}$ followed by a unitary transformation $\mathcal{U}$ \textit{i.e.} $\mathcal{E} = \mathcal{U}\mathcal{K}$. Though unitary transformation $\mathcal{U}$ is a
valid physical reversible evolution consistent with both RC and AC, however, implementing any such $\mathcal{E}$ includes implementing $\mathcal{K}$ followed by the corresponding $\mathcal{U}$. However, as established in theorem $1$, perfect implementation of $\mathcal{K}$ contradicts AC (though perfectly consistent with RC). Hence, $\mathcal{E}$ cannot be a valid reversible evolution in any generic (as well as the standard ) formulation of quantum theory where AC is respected. \end{proof} 

\textit{POPT states and allowed reversible operations} :--
A set of positive operators $\{\pi^a_A\}_a$ satisfying $\sum_a\pi^a_A=\mathbb{I}_A$ constitute a generic POVM measurement; where $\forall~a,~\pi^a_A\in\mathcal{E}(\mathcal{H}_A)$ with $\mathcal{E}(\mathcal{H}_A)$ ( Effect space ) denotes set of all positive operators on $\mathcal{H}_A$; and $\mathbb{I}_A$ is the identity operator on  $\mathcal{H}_A$. The probability $p(a|M_A)$ that Alice obtains an outcome $a$ for measurement $M_A\equiv\{\pi^a_A\}$ is given by a probability measure $\mu:\mathcal{E}(\mathcal{H}_A)\mapsto[0,1]$ that satisfies the following properties : (i) $\forall~\pi^a_A\in\mathcal{E}(\mathcal{H}_A),~0\le\mu(\pi^a_A)\le1,~$ (ii) $\mu(\mathbb{I}_A)=1$, and (iii) $\mu(\sum_i\pi^i_A)=\sum_i\mu(\pi^i_A)$ with $\sum_i\pi^i_A\le \mathbb{I}_A$. The celebrated Gleason-Busch theorem tells that any such generalized probability measure is of the
form $\mu(\pi^a_A)=\tr(\rho_A\pi^a_A)$, for some density operator $\rho_A\in\mathcal{D}(\mathcal{H}_A)$ \cite{Gleason57,Busch03}. Consider now a bipartite system with Hilbert space $\mathcal{H}_A \otimes \mathcal{H}_B$ but restrict the bipartite effect space only to product effects \cite{Klay87,Wallach00,Barnum05}. It has been shown that any probability measure on such product effects, {\it i.e.} $\mu:\mathcal{E}(\mathcal{H}_{A}) \times \mathcal{E}(\mathcal{H}_{B}) \mapsto[0,1]$ satisfying the conditions (i)-(iii) is of the form $\mu(\pi^{a}_B \otimes\pi^{b}_{B})=\tr[W(\pi^{a}_{A}\otimes\pi^{b}_{B})]$ for some Hermitian operator $W$ acting on $\mathcal{H}_{A} \otimes \mathcal{H}_{B}$. Clearly, $W$ is positive over all pure tensors (POPT). Importantly, the set of all POPTs on  $\mathcal{H}_A \otimes \mathcal{H}_B$ is a strict superset of the set of all density operators on $\mathcal{H}_A \otimes \mathcal{H}_B$. However, one can consider generalised toy-theory with local state space quantum but bipartite state space extended upto 'maximal tensor-product' which includes all POPTs as valid bipartite states. \cite{jan,jan1}. Now notice that $ \mathcal{S} : = \mathcal{I}\otimes\mathcal{T}(\ket{\psi^-}\bra{\psi^-})\in \mbox{Herm}\left(\mathbb{C}^2\otimes\mathbb{C}^2\right)$; $\ket{\psi^-}:=(\ket{01}-\ket{10})/\sqrt{2}$ and $\mbox{Herm}(\mathcal{H})$ denotes the set of hermitian operators acting on $\mathcal{H}$ and $\mathcal{T}$ denotes transposition. Straightforward calculation yields negative eigenvalue for $\mathcal{S}$ prohibiting it to be a density operator. However, it is also straightforward to see that $\mathcal{S}$ yields non-negative probabilities on any
product effect and hence it is a POPT state. However, what is crucial is the following fact : $Tr(\mathcal{S} \sigma_z \otimes \sigma_z) = Tr(\mathcal{S} \sigma_x \otimes \sigma_x) = -1 = - Tr(\mathcal{S}\sigma_y \otimes \sigma_y) $. Therefore, $S$ yields anti-correlated outcomes when Alice and Bob makes $\sigma_z$ or $\sigma_x$ measurement and yields anti-correlated outcome of they measure $\sigma_y$. Therefore, if we allow $\mathcal{S}$ as a valid bipartite state, then the proof of Theorem $1$ breaks down since $\mathcal{S}$ serves as a common intersection point in convex hulls of all three scenarios. Therefore, the incompatibility of anti-linear anti-unitary logic with AC depends on what we postulate as the global state space of bipartite system. If we go beyond the quantum bipartite state space and postulate the post-quantum 'maximal tensor product' as the bipartite state space and thereby allow states with all possible combinations of correlations, then AC does not rule out the possibility of anti unitary logic gate. 

{\it Discussions:--} We have derived a novel connection between the quantum mechanical restriction on possible class of reversible manipulations of qubits and the causal structure of spacetime, based on the allowed class of correlations between local observables in valid bipartite state. First we have shown that, although many existing dynamical impossibility ( e.g. no-cloning) can be derived as a consequence of relativistic causality, this is however not the case for reversible quantum gates that are antilinear anti-unitary transformations. In other words, perfect implementation of antilinear anti-unitary gates does not lead to a violation of relativistic causality. Therefore, we go beyond relativistic causality and come up with a generic causal primitive namely the principle of Absoluteness of Cause. We have derived a proof that any formulation of quantum theory where time orders of two spacelike separated local operations does not cause distinct observations in their causal future (thereby respecting the principle of AC), perfect antilinear anti-unitary logic gates on a quantum bit cannot exist, even though these logic gates are perfectly reversible similar to their linear unitary counterpart. Finally, it has been shown that if we go beyond quantum bipartite state space to post-quantum maximal tensor product state space that includes POPTs as valid bipartite state, the contradiction between AC and antilinear gates dissolves. This reveals the intricate connection between allowed class of correlation present in bipartite (or in general, multipartite) quantum states and class of allowed local manipulation of qubits, which can be further investigated.

ADB acknowledges insightful discussions with Manik Banik and Guruprasad Kar.

\newpage
\onecolumngrid
\begin{center}
   \large {{\bf Supplemental: Reversible computation and the causal structure of space-time}}
\end{center}

\section{Violation of AC does not necessarily imply violation of RC}

Consider Alice and Bob share a bipartite quantum state $\rho_{AB}$. Consider two spacelike separated quantum operations $\Lambda_A : \mathcal{D(H_A)}\xrightarrow{} \mathcal{D(H'_A)}$ and $\Lambda_B : \mathcal{D(H_B)}\xrightarrow{} \mathcal{D(H'_B)}$ on system A and B respectively. Alice performs $\Lambda_A$ on her system $A$, Bob performs $\Lambda_B$ on $B$. The time order of these two spacelike separated operations is reference frame dependent.
The principle of AC prohibits the possible time orders
of such spacelike separated operations to be a potential
cause of distinct observations at their common future.

Consider two possible time order of Alice and Bob's local operations.
\begin{itemize}
\item[1.] Alice performing $\Lambda_A$ on her system precedes Bob performing $\Lambda_B$ on his system. The joint state $\rho_{AB}$ will evolve to  $\tau_{AB}$.
\item[2.] Bob performing $\Lambda_B$ on his part precedes Alice performing $\Lambda_A$ on her part. In this case, the joint state $\rho_{AB}$ evolves to  $\tau'_{AB}$
\end{itemize}

If $\tau_{AB} \neq \tau'_{AB}$, AC would be violated. (since, according to AC, time orders of spacelike operations cannot be the cause for distinct future observations). The violation, however, could happen in the following two ways.

\begin{itemize}
\item[(a)] $\tau_{AB}$ and $\tau'_{AB} $ yield different marginal statistics \textit{i.e.} $Tr_A(\tau_{AB}) \neq Tr_A(\tau'_{AB}) $. Such a violation enables instantaneous communication from Alice to Bob. Therefore thus contradicts relativistic causality. Such a violation of AC implies violation of RC.

\item[(b)] However, $\tau_{AB}$ can differ from $\tau'_{AB}$ in a way that both yield identical marginal statistics but still  $\tau_{AB} \neq \tau'_{AB}$ \textit{i.e.} $\tau_{AB}$ and $\tau'_{AB}$
cannot be distinguished neither by Alice nor by Bob
locally. The states can be distinguished only if Alice and Bob come together in their common future and compare their local statistics or perform some kind of global measurement which cannot be locally realisable. Such a violation of AC does not leads to violation of RC, since, unlike the previous case in (a), both the parties need to come together at their causal future to observe this violation.
\end{itemize}

However, neither (a) nor (b) is allowed by quantum theory since (as already discussed in the manuscript), quantum theory ensures $\tau_{AB} = \tau'_{AB}$

\newpage

\section{Different time Orders of spacelike separated operations  }
\begin{figure}[h!]
\centering
\includegraphics[width=.97\textwidth]{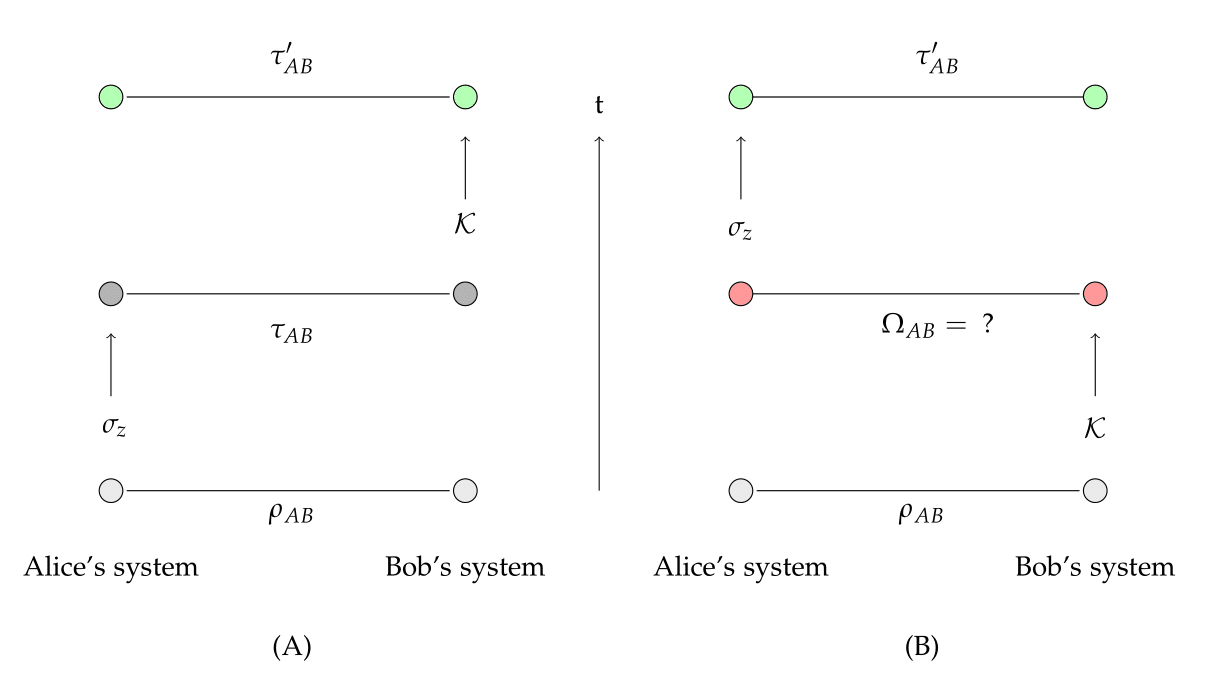}
\caption{(Color online) The figure shows that different possible time orders of two spacelike separated operations yield exactly the same final state $\tau'_{AB}$. Therefore, these time-orderings of spacelike separated operations cannot be the cause for distinct future observations, which is in accordance with the principle of AC (since such time-orderings of spacelike operations are reference frame dependent and hence are not absolute and therefore cannot be the cause for events in their causal future). The figure considers only the case of scenario $1$ in proof of Theorem $1$ where Alice's local operation is $\sigma_z$ measurement on her system and Bob's local measurement is applying complex conjugation gate $\mathcal{K}$ on his system. 
Part \textbf{(A)} considers the time order in which Alice's $\sigma_z$ measurement precedes Bob's application of $\mathcal{K}$. Due to Alice's local operation, the initial bipartite state $\rho_{AB}$ transforms into $\tau_{AB}$. Then due to application of $\mathcal{K}$ by Bob on his system, $\tau_{AB}$ transforms into $\tau'_{AB}$. Therefore, after completion of both of their local operations the final joint state is $\tau'_{AB} $. Part \textbf{(B)} considers the reverse time order where Bob's application $\mathcal{K}$ precedes Alice's $\sigma_z$ measurement. Since we do not assume Schrodinger-type dynamical linearity on superposition of quantum states, the possibility for the evolved joint state $\Omega_{AB}$ is large. However, due to AC principle, there is one crucial condition on $\Omega_{AB}$ which is following : Alice's $\sigma_z$ measurement on $\Omega_{AB}$ must lead to $\tau'_{AB}$ \textit{i.e.} this reverse time ordering of Alice and Bob's operations should lead to exactly the same final state as in the other time ordering of the concerned local operations.}\label{fig1}
\end{figure}

Scenario $2$ and $3$ in the proof of Theorem $1$ also admit similar schematic diagrams.


\begin{thebibliography}{99}
\bibitem{note} Any computational circuit containing irreversible gates can be replaced by an equivalent circuit containing only reversible gates by making use of the reversible Toffoli gate.
\end{thebibliography}
\end{document}